\documentclass[a4paper,UKenglish,cleveref]{lipics-v2021}

\pdfoutput=1 %uncomment to ensure pdflatex processing (mandatatory e.g. to submit to arXiv)
\hideLIPIcs  %uncomment to remove references to LIPIcs series (logo, DOI, ...), e.g. when preparing a pre-final version to be uploaded to arXiv or another public repository

\nolinenumbers

\title{Formalising Yoneda Ext in Univalent Foundations}

\author{Jarl G. Taxerås {Flaten}}
{University of Western Ontario, London, Ontario, Canada}
{jtaxers@uwo.ca}
{}
{}

\authorrunning{J. G. T. Flaten}
\Copyright{Jarl G. Taxerås Flaten}



\relatedversion{}
\relatedversiondetails
{Full version}
{https://arxiv.org/abs/2302.12678}

\supplement{The source code of the formalisation described in this paper is found at:}
\supplementdetails[subcategory={Source Code}]
{Software}
{https://github.com/HoTT/HoTT/tree/master/theories/Algebra/AbSES}
\supplementdetails[subcategory={Source Code}]
{Software}
{https://github.com/jarlg/Yoneda-Ext}

\funding{We acknowledge the support of the Natural Sciences and Engineering Research Council of Canada (NSERC), RGPIN-2022-04739.}

\acknowledgements{The theory of Yoneda Ext in HoTT is joint work with Dan Christensen, to whom I am also grateful for discussions and contributions related to the formalisation.
  I thank Jacob Ender for contributions to the formalisation of the Baer sum, and the collaborators of the Coq-HoTT library---and Ali Caglayan in particular---for their review of, and contributions to, the various pull requests related to this project.
  I also thank the anonymous reviewers for valuable feedback.}

\EventEditors{Adam Naumowicz and Ren\'{e} Thiemann}
\EventNoEds{2}
\EventLongTitle{14th International Conference on Interactive Theorem Proving (ITP 2023)}
\EventShortTitle{ITP 2023}
\EventAcronym{ITP}
\EventYear{2023}
\EventDate{July 31 to August 4, 2023}
\EventLocation{Bia\l{}ystok, Poland}
\EventLogo{}
\SeriesVolume{268}
\ArticleNo{16}

\usepackage{tikz-cd}

\DeclareMathOperator{\Ab}{Ab}
\DeclareMathOperator{\Ext}{Ext}
\DeclareMathOperator{\Zb}{\mathbb{Z}}
\DeclareMathOperator{\Nb}{\mathbb{N}}
\DeclareMathOperator{\Type}{\verb~Type~}
\DeclareMathOperator{\splice}{\circledcirc}
\DeclareMathOperator{\vExt}{{\verb~Ext~}}
\DeclareMathOperator{\vES}{{\verb~ES~}}
\DeclareMathOperator{\vAbSES}{{\verb~AbSES~}}
\DeclareMathOperator{\vHom}{{\verb~Hom~}}

\DeclareMathOperator{\pType}{{\verb~pType~}}
\DeclareMathOperator{\id}{id}
\DeclareMathOperator{\pt}{\verb~pt~}
\DeclareMathOperator{\op}{op}

\newcommand\fib[1]{{\verb~fib~}_{#1}}

\newcommand\mtt{\mathtt}
\newcommand\jeq{\equiv}
\newcommand\tr[1]{\lvert #1 \lvert}
\newcommand\Tr[1]{\lVert #1 \lVert } % truncation of a type

\newcommand\xra{\xrightarrow}

\newcommand\coderef{\diamondsuit}

\newcommand{\period}{\rlap{\hspace{2pt}.}}

\begin{document}

\maketitle

\begin{abstract}
  Ext groups are fundamental objects from homological algebra which underlie important computations in homotopy theory.
  We formalise the theory of Yoneda Ext groups~\cite{Yoneda1954} in homotopy type theory (HoTT) using the Coq-HoTT library~\cite{coqhott}.
  This is an approach to Ext which does not require projective or injective resolutions, though it produces large abelian groups.
  Using univalence, we show how these Ext groups can be naturally represented in HoTT.
  We give a novel proof and formalisation of the usual six-term exact sequence via a fibre sequence of \(1\)-types (or groupoids), along with an application.
  In addition, we discuss our formalisation of the contravariant long exact sequence of Ext, an important computational tool.
  Along the way we implement and explain the Baer sum of extensions and how Ext is a bifunctor.
\end{abstract}

\section{Introduction}

The field of homotopy type theory (HoTT) lies at the intersection of type theory and algebraic topology, and serves as a bridge to transfer tools and insights from one domain to the other.
In one direction, the formalism of type theory has proven to be a powerful language for reasoning about some of the highly coherent structures occurring in branches of modern algebraic topology.
Several of these structures are ``natively supported'' by HoTT, and we can reason about them much more directly than in classical set-based approaches.
This makes HoTT an ideal language in which to formalise results and structures from algebraic topology.
Moreover, theorems in HoTT are valid in any \emph{\(\infty\)-topos}, not just for ordinary spaces.
Details about the interpretation of our constructions into an \(\infty\)-topos, and the relation of our Ext groups to \emph{sheaf Ext}, are discussed in~\cite{CF}.

We present a formalisation of Ext groups in HoTT following the approach of Yoneda~\cite{Yoneda1954, Yoneda1960}.
Ext groups are fundamental objects in homological algebra, and they permeate computations in homotopy theory.
For example, the universal coefficient theorem relates Ext groups and cohomology, and features in the classical proof that \(\pi_5(S^3) \simeq \Zb / 2\).
Much of our formalisation has already been accepted into the Coq-HoTT library under the \href{https://github.com/HoTT/Coq-HoTT/tree/56629c19010a7d7155b22aad4525f9b2ac1bd584/theories/Algebra/AbSES}{\verb~Algebra.AbSES~} namespace, though we have also contributed to other parts of the library throughout this project.
The long exact sequence, along with a few other results we need, are currently in a separate repository named \href{https://github.com/jarlg/Yoneda-Ext}{\verb~Yoneda-Ext~}.
We supply links to formalised statements using a trailing \(^\coderef\)-sign throughout.

In ordinary mathematics, Ext groups of modules over a ring are usually defined using projective (or injective) resolutions.
This is possible because the axiom of choice implies the existence of such projective resolutions, and Ext groups are independent of any particular choice of resolution.
(Similarly, categories of sheaves of modules always admit injective resolutions.)
In our setting, however, even abelian groups fail to admit projective resolutions.
This stems from the fact that some sets fail to be projective, which may be familiar to those working constructively or internally to a topos.
Accordingly, to define Ext groups in homotopy type theory we cannot rely on resolutions.
Fortunately, Yoneda \cite{Yoneda1954, Yoneda1960} gave such a general approach, whose theory is detailed in \cite{Mac63}, our main reference.
A drawback of this approach is that it produces \emph{large} abelian groups, as we explain in \cref{ssec:ses}.

We build upon the Coq-HoTT library~\cite{coqhott}, which contains sophisticated homotopy-theoretic results, but which is presently lacking in terms of ``basic'' algebra.
For this reason, we have opted to simply develop Ext groups of abelian groups, instead of for modules over a ring or in a more general setup.
Nevertheless, it is clear that everything we do could have been done over an arbitrary ring, given a well-developed library of module theory.
Moreover, we emphasise that higher Ext groups in HoTT are interesting even for abelian groups.
While in classical mathematics such Ext groups of abelian groups are trivial in dimension \(2\) and up, in HoTT they may be nontrivial in all dimensions!
This is because there are models of HoTT in which these Ext groups are nontrivial~\cite{CF}.

In \cref{sec:yext} we explain how univalence lets us naturally represent Yoneda's approach to Ext in HoTT.
We construct the type \(\vAbSES(B, A)\) of short exact sequences between two abelian groups \(A\) and \(B\), and define \(\vExt^1(B, A)\) to be the set of path-components of \(\vAbSES(B, A)\).
This definition is justified by characterising the paths in \(\verb+AbSES+(B, A)\), which crucially uses univalence.
We also show that the loop space of \(\verb+AbSES+(B, A)\) is isomorphic to the group \(\verb+Hom+(B,A)\) of group homomorphisms, and that \(\verb+Ext+^1(P,A)\) vanishes whenever \(P\) is projective, in a sense we define.
These results all play a role in the subsequent sections.

The main content of \cref{sec:fibre-sequence} is a proof and formalisation of the following:

\theoremstyle{plain}
\newtheorem*{thm:pullback-fibre-sequence}{\cref{thm:pullback-fibre-sequence}}
\begin{thm:pullback-fibre-sequence}
  Let \(A \xra{i} E \xra{p} B\) be a short exact sequence of abelian groups.
  For any abelian group \(G\), pullback yields a fibre sequence:
  \( \vAbSES(B, G) \xra{p^*} \vAbSES(E, G) \xra{i^*} \vAbSES(A, G). \)\href{https://github.com/HoTT/Coq-HoTT/blob/56629c19010a7d7155b22aad4525f9b2ac1bd584/theories/Algebra/AbSES/PullbackFiberSequence.v\#L390}{\(^\coderef\)}
\end{thm:pullback-fibre-sequence}

We give a novel, direct proof of this result which requires managing considerable amounts of coherence.
The formalisation is done for abelian groups, but the proof applies to modules over a general ring.
Its formalisation benefited from the \verb+WildCat+ library of Coq-HoTT (see \cref{p:wildcat}), which makes it convenient to work with types equipped with an imposed notion of paths.
This allows us to work with \emph{path data} in \(\vAbSES(B, A)\) with better computational properties than actual paths, but which correspond to paths via the aforementioned characterisation.
From the fibre sequence of the theorem we deduce the usual six-term exact sequence (\cref{prop:six-term}), which we then use to compute Ext groups of cyclic groups:
\[ \verb+Ext+^1(\Zb/n, A) \ \cong \ A/n \]
for any nonzero \( n : \Nb \) and abelian group \(A\) (\cref{cor:ext-cyclic}).%
\href{https://github.com/HoTT/Coq-HoTT/blob/832aef3e6fff0f5b953ed170522e1a3d6288a4bb/theories/Algebra/AbSES/SixTerm.v\#L222}{\(^\coderef\)}
The six-term exact sequence, along with this corollary, have already been applied in~\cite{BCFR}.
We also discuss how \(\vExt^1\) becomes a bifunctor into abelian groups using the Baer sum.

Finally, in \cref{sec:les} we define \(\verb+Ext+^n\) for any \(n : \Nb\) and discuss our formalisation of the long exact sequence, in which the connecting maps are given by \emph{splicing}:%
\href{httphttps://github.com/jarlg/Yoneda-Ext/blob/0d8bfe8e168bbdf325e805d7268e826e889189f0/LES.v\#L10}{\(^\coderef\)}

\newtheorem*{thm:les}{\cref{thm:les}}
\begin{thm:les}
  Let \(A \xra{i} E \xra{p} B\) be a short exact sequence of abelian groups.
  For any abelian group \(G\), there is a long exact sequence by pulling back:%
  \(\href{https://github.com/jarlg/Yoneda-Ext/blob/0d8bfe8e168bbdf325e805d7268e826e889189f0/LES.v\#L15}{^\coderef}
  \href{https://github.com/jarlg/Yoneda-Ext/blob/0d8bfe8e168bbdf325e805d7268e826e889189f0/LES.v\#L47}{^\coderef}
  \href{https://github.com/jarlg/Yoneda-Ext/blob/0d8bfe8e168bbdf325e805d7268e826e889189f0/LES.v\#L94}{^\coderef}\)
  \[ \cdots \xra{i^*} \verb+Ext+^n(A,G) \xra{- \splice E} \verb+Ext+^{n+1}(B,G) \xra{p^*} \verb+Ext+^{n+1}(E,G) \xra{i^*} \cdots \period \]
\end{thm:les}

At present, we have only formalised this long exact sequence of \emph{pointed sets}.
It remains to construct the Baer sum making \(\vExt^n\) into an abelian group for \( n > 1\), however once this is done then we automatically get a long exact sequence of abelian groups.
Our proof follows that of Theorem~5.1 in \cite{Mac63}, which is originally due to Stephen Schanuel.

\subparagraph*{Notation and conventions}
We use typewriter font for concepts which are defined in the code, such as \verb+AbSES+ and \verb+Ext+.
In contrast, when we use normal mathematical font, such as \(\Ext^n(B,A)\), we mean the classical notion.
For mathematical statements we prefer to stay close to mathematical notation by
writing for example \(\verb+Ext+^n(B,A)\) for what means \(\verb+Ext n B A+\) in Coq.
The symbol \(^\coderef\) is used to refer to relevant parts of the code.

Our terminology mirrors that of~\cite{hottbook}; in particular we say `path types' for what are also called `identity types' or `equality types'.
We write \(\pType\) for the universe of pointed types, and \(\pt\) for the base point of a pointed type.
The \(\jeq\)-symbol is for definitional equality.

\section{Preliminaries}

\subsection{Homotopy Type Theory}
We briefly explain the formal setup of homotopy type theory along with some basic notions that we need.
For a thorough introduction to HoTT, the reader may consult~\cite{hottbook, rijke-intro}.

Homotopy type theory (HoTT) extends Martin--L\"of type theory (MLTT) with the \emph{univalence axiom} and often various higher inductive types (HITs).
Of the latter, we simply need propositional truncation and set truncation, which we explain in more detail below.

The univalence axiom characterises the identity types of universes.
In ordinary MLTT, there is always a function
\[ \verb+idtoequiv+ : \prod_{X, Y : \Type} (X = Y) \to (X \simeq Y) \]
defined by sending the reflexivity path on a type \(X\) to the identity self-equivalence on \(X\), using the induction principle of path types.
The univalence axiom asserts that \verb+idtoequiv+ is an equivalence for all \(X\) and \(Y\).
In HoTT, the first thing we often do after defining a new type is to characterise its path types.
The univalence axiom does this for the universe.

From univalence, a general \emph{structure identity principle}~\cite[Chapter~9.8]{hottbook} follows which characterises paths between structured types, such as groups and other algebraic structures.
In the case of groups, univalence implies that paths between groups correspond to group isomorphisms.
Similarly, paths between modules correspond to module isomorphisms.

\paragraph*{Propositions, sets, and groupoids}
In HoTT there is a hierarchy of \textbf{\boldmath{\(n\)}-truncated types} (or \textbf{\boldmath{\(n\)}-types}, for short) for any integer \(n \geq -2\).
In general, a type \(X\) is an \((n+1)\)-type when all the path types \(x_0 =_X x_1\) are \(n\)-types.
The recursion starts at \(-2\), when the condition is just that the map \(X \to 1\) is an equivalence, and in this case \(X\) is \textbf{contractible}.

We only deal with the bottom four levels of this hierarchy: \textbf{contractible types}, \textbf{propositions} (\((-1)\)-types), \textbf{sets} (0-types) and 1-types.
A type \(X\) is a proposition when any two points in \(X\) are equal (but there may not be any points).
A type \(X\) is a set when the path types \(x_0 =_X x_1\) are all propositions---this amounts to there being ``at most'' one path between \(x_0\) and \(x_1\).
Lastly, a type \(X\) is a 1-type when its path types are sets---in particular, for any \(x : X\), the \textbf{loop space} \(\Omega X :\jeq (x =_X x)\) is a set which is a group under path composition.
(We leave base points implicit when taking loop spaces.)

There are truncation operations which create a proposition or a set from a given type~\(X\).
We denote by \(\Tr{X}\) the propositional truncation, and by \(\pi_0{X}\) the set truncation (or \emph{set of path-components}) of \(X\).
In Coq-HoTT, the corresponding notation is \verb~merely X~ and \verb~Tr 0 X~.
The map \(\verb~tr~ : X \to \pi_0X\) sends a point to its connected component.
When we say that a type \(X\) \textbf{merely} holds, then we mean that its propositional truncation \(\Tr{X}\) holds.

\subsection{The Coq-HoTT Library}

The Coq-HoTT library~\cite{coqhott} is an open-source repository of formalised mathematics in homotopy type theory using Coq.
It is particularly aimed at developing synthetic homotopy theory, and includes theory about spheres, loop spaces, classifying spaces, modalities, ``wild \(\infty\)-categories,'' and basic results about abelian groups, to mention a few things.
The library is part of the \emph{Coq Platform} and is available through the standard \verb+opam+ package repositories.

Below we explain some of the main features of this library, and of Coq itself, which are important for the present work.

\paragraph*{Universes and cumulativity}
We assume basic familiarity with universes and universe levels in Coq, and in particular that they are \emph{cumulative}: a type \(X : \verb+Type@{u}+\) can be resized to live in \verb+Type@{v}+ under the constraint \(\verb+u+ \leq \verb+v+\).
(Here \verb+u+ and \verb+v+ are \textbf{universe levels}.)
Resizing is done implicitly by Coq.

In the Coq-HoTT library, we additionally make most of our structures cumulative.
This essentially means that resizing commutes with the formation of a data structure---i.e., it does not matter whether you resize the inputs to the data structure or whether you resize the resulting data structure.
As an example, consider the data structure \verb+prod+ which forms the product of two types in a common (for simplicity) universe level.
Suppose we have two universe levels \verb+u+ and \verb+v+ with the constraint \verb+u < v+.
Given \verb+X Y : Type@{u}+, we can form the product at level \verb~u~ and then resize, or first resize and then form the product.
By making \verb+prod+ a cumulative data structure, the two results agree (with implicit resizing):
\[ \verb+prod@{u} X Y + \jeq \verb+ prod@{v} X Y+. \]

Cumulativity of data structures is an essential Coq feature which facilitates the kind of formalisation we do in this paper.
For example, it lets us resize groups and homomorphisms.
It also lets us reduce the number of universes in some of our definitions via the following trick: instead of having separate universes for different inputs, we can often use a single universe (which represents the maximum) and leverage cumulativity.

We also make use of universe constraints since our constructions move between various universe levels.
The constraints both document and verify the mathematical intent.

\paragraph*{The \verb~WildCat~ library} \label{p:wildcat}

The \href{https://github.com/HoTT/Coq-HoTT/tree/56629c19010a7d7155b22aad4525f9b2ac1bd584/theories/WildCat}{\verb~WildCat~} namespace contains the development of ``wild \(\infty\)-categories,'' functors between such, and related things.
This library was spearheaded by Ali Caglayan, tslil clingman, Floris van Doorn, Morgan Opie, Mike Shulman, and Emily Riehl.
The concepts generalise those appearing in \cite[Section~4.3.1]{vanDoorn2018}, and are not currently present in the literature.
We explain the basics of this library which are especially relevant for our formalisation.

Starting from the notion of \textbf{graph}\href{https://github.com/HoTT/Coq-HoTT/blob/56629c19010a7d7155b22aad4525f9b2ac1bd584/theories/WildCat/Core.v\#L9}{\(^\coderef\)}---a type \(A\) with a binary operation (or \emph{correspondence}) \(\vHom\) into \(\Type\)---the notion of a \textbf{0-functor}\href{https://github.com/HoTT/Coq-HoTT/blob/56629c19010a7d7155b22aad4525f9b2ac1bd584/theories/WildCat/Core.v\#L84}{\(^\coderef\)} is that of a homomorphism of graphs:

\begin{lstlisting}
Class IsGraph (A : Type) := { Hom : A -> A -> Type }.
Class Is0Functor {A B : Type} `{IsGraph A} `{IsGraph B} (F : A -> B)
  := { fmap : forall {a b : A} (f : Hom a b), Hom (F a) (F b) }.
\end{lstlisting}
We will often use the notation \(\vHom\) in this text, leaving the graph structure implicit.

From here one could go ahead and define categories by defining a composition operation and using the identity types of the type \(\vHom(a,b)\) to express the various laws a category needs to satisfy, such as associativity of composition.
A more flexible approach is to instead allow \(\vHom(a,b)\) to itself be a graph, making \(A\) into a \textbf{2-graph}.\href{https://github.com/HoTT/Coq-HoTT/blob/56629c19010a7d7155b22aad4525f9b2ac1bd584/theories/WildCat/Core.v\#L89}{\(^\coderef\)}
This is the approach taken by \verb~WildCat~, and this flexibility is important for our formalisation.

\begin{lstlisting}
Class Is2Graph (A : Type) `{IsGraph A}
  := { isgraph_hom : forall (a b : A), IsGraph (Hom a b) }.
\end{lstlisting}

For a \(2\)-graph \(A\), a category structure can then be defined in a straightforward manner using \verb~isgraph_hom~ to express the various laws that need to hold.
This structure is bundled into a class called \verb+Is1Cat+.\href{https://github.com/HoTT/Coq-HoTT/blob/56629c19010a7d7155b22aad4525f9b2ac1bd584/theories/WildCat/Core.v\#L95}{\(^\coderef\)}
For example, associativity is expressed as follows, using the notation \verb+$==+ as a shorthand for the \(2\)-graph structure and \verb~$o~ for composition:
\begin{lstlisting}
  cat_assoc : forall (a b c d : A)
    (f : Hom a b) (g : Hom b c) (h : Hom c d),
      (h $o g) $o f $== h $o (g $o f);
\end{lstlisting}

If all the morphisms in \(A\) are invertible, then \(A\) is a \textbf{groupoid}.%
\href{https://github.com/HoTT/Coq-HoTT/blob/56629c19010a7d7155b22aad4525f9b2ac1bd584/theories/WildCat/Core.v\#L367}{\(^\coderef\)}
Finally, for the notion of a \textbf{\boldmath{\(1\)}-functor} between categories we also express the laws using the \(2\)-graph structure.%
\href{https://github.com/HoTT/Coq-HoTT/blob/56629c19010a7d7155b22aad4525f9b2ac1bd584/theories/WildCat/Core.v\#L252}{\(^\coderef\)}

\begin{lstlisting}
Class Is1Functor {A B : Type} `{Is1Cat A} `{Is1Cat B}
  (F : A -> B) `{!Is0Functor F} := {
    fmap_id : forall a, fmap F (Id a) $== Id (F a);
    fmap_comp : forall a b c (f : Hom a b) (g : Hom b c),
      fmap F (g $o f) $== fmap F g $o fmap F f;
    fmap2 : forall a b (f g : Hom a b),
      (f $== g) -> (fmap F f $== fmap F g) }.
\end{lstlisting}

The terms \verb+fmap_id+ and \verb+fmap_comp+ express that the functor \verb+F+ respects identities and composition, as usual.
If we had used identity types instead of a \(2\)-graph structure, so that \verb+f $== g+ simply meant \verb+f = g+, then \verb+F+ would automatically respect equality between morphisms, making \verb+fmap2+ redundant.
However, in the more general \(2\)-graph setup, this needs to be included as a law.

  The adjective ``wild'' is used for the sort of categories just defined to indicate that they do not capture all the coherence needed to represent \(\infty\)-categories, only the \(1\)-categorical structure.
  However, in our usage we will only encounter genuine \(1\)-categories and groupoids.
  In particular, any type \(X\) defines a groupoid via its identity types\href{https://github.com/HoTT/Coq-HoTT/blob/56629c19010a7d7155b22aad4525f9b2ac1bd584/theories/WildCat/Paths.v\#L7}{\(^\coderef\)}, and if \(X\) is a \(1\)-type then this groupoid structure captures everything about \(X\).
  This enables us to impose our own notion of paths, which we call \emph{path data} below, for certain types of interest.

\section{Yoneda Ext} \label{sec:yext}

As mentioned in the introduction, we will follow Yoneda's approach to Ext groups~\cite{Yoneda1954, Yoneda1960}, which does not require projective (or injective) resolutions, though it produces \emph{large} groups.
This approach and related theory is explained in~\cite{Mac63}, which is our main reference.
At present, the Coq-HoTT library---with which this work has been formalised---does not contain much theory related to modules over a general ring (nor the theory of abelian categories, or anything of the sort).
We therefore only formalise and state our results for abelian groups.
It is clear, however, that everything we say could be done for modules over a general ring.

For the classically-minded reader, let us also emphasise that in homotopy type theory the category of abelian groups does \emph{not} have global dimension \(1\), so that the higher Ext groups we define in \cref{sec:les} do not necessarily vanish.

\subsection{The Type of Short Exact Sequences} \label{ssec:ses}

Given two abelian groups \(A\) and \(B\), Yoneda defines a group \(\Ext^1(B,A)\) by considering the large set (or class) of all short exact sequences
\( A \xra{i} E \xra{p} B \) and taking a quotient by a certain equivalence relation.
The sequence being exact means that \(i\) is injective, \(p\) is surjective, that \(p \circ i = 0\), and that the image of \(i\) is equal to the kernel of \(p\).
We usually simply write \(E\) for the short exact sequence \(A \to E \to B\) when no confusion can arise.
The equivalence relation which Yoneda quotients out by is defined as
``\( E \sim F \) if and only if there exists an isomorphism \(E \cong F\) which respects the maps from \(A\) and to \(B\).''
Equivalently, but more topologically, one can consider the \emph{groupoid} of short exact sequences \( A \to E \to B \) and define \(\Ext^1(B,A)\) to be the set of path-components of this groupoid---see, e.g., \cite[Chapter~III]{Mac63} for details about both of these descriptions.

In homotopy type theory, given two abelian groups \(A\) and \(B\) we form the \textbf{type of short exact sequences} from \(A\) to \(B\) as the \(\Sigma\)-type over all abelian groups \(E\) equipped with an injection \(\mtt{inclusion}_E : A \to E\), a surjection \( {\verb+projection+}_E : E \to B\), and a witness that these two maps form an exact complex.
We represent this data as the following record-type:\href{https://github.com/HoTT/Coq-HoTT/blob/3062f0a152dca5b58323bffb9fceab4188d96bb1/theories/Algebra/AbSES/Core.v\#L19}{\(^\coderef\)}

\begin{lstlisting}[label=list:abses]
Record AbSES@{u v | u < v} (B A : AbGroup@{u}) : Type@{v} := {
    middle : AbGroup@{u};
    inclusion : Hom A middle;
    projection : Hom middle B;
    isembedding_inclusion : IsEmbedding inclusion;
    issurjection_projection : IsSurjection projection;
    isexact_inclusion_projection
      : IsExact (Tr (-1)) inclusion projection;
  }.
\end{lstlisting}

Note that \(\vAbSES(B,A)\) denotes short exact sequences \emph{from \(A\) to \(B\)}.
%(The order is chosen for reasons of variance.)
The abelian group \verb+middle+ plays the role of \(E\) in the prose above.
Here, the condition that \( \verb+projection+_E \circ \verb+inclusion+_E = 0\) is baked into the \verb+IsExact+ field, which also expresses exactness.\footnote{The term \verb~Tr (-1)~ can safely be ignored; it expresses that the induced map from \(A\) to the kernel of \(\verb~projection~_E\) is \emph{\((-1)\)-connected}, which here just means it is a surjection.}
We have included universe annotations which express that \(E\) lives in the same universe \verb+u+ as the abelian groups \(A\) and \(B\).
Accordingly, the resulting type \verb+AbSES+\((B,A)\) lives in a universe \verb+v+ which is strictly greater than \verb+u+, as in Yoneda's construction above.
The type \verb+AbSES+\((B,A)\) is pointed by the \textbf{trivial short exact sequence}\href{https://github.com/HoTT/Coq-HoTT/blob/3062f0a152dca5b58323bffb9fceab4188d96bb1/theories/Algebra/AbSES/Core.v\#L60}{\(^\coderef\)}
\(A \to A \oplus B \to B.\)

We now define \(\vExt^1(B,A)\) as the set-truncation of the type of short exact sequences.%
\href{https://github.com/HoTT/Coq-HoTT/blob/3062f0a152dca5b58323bffb9fceab4188d96bb1/theories/Algebra/AbSES/Ext.v\#L21}{\(^\coderef\)}

\begin{lstlisting}[label=list:ext]
Definition Ext (B A : AbGroup) := Tr 0 (AbSES B A).
\end{lstlisting}

In \cref{ssec:baer-sum} we make the set \(\vExt^1(B,A)\) into an abelian group via the \emph{Baer sum}.
These abelian groups, and their higher variants defined in \cref{sec:les}, are our main objects of study.

Whenever we define a new type in homotopy type theory, the first thing we often do is to characterise its path types.
Theorem~7.3.12 of \cite{hottbook} characterises paths in truncations, yielding
\[ \big( \tr{E}_0 =_{\Ext^1} \tr{F}_0 \big) \ \simeq \ \Tr{E = F} \]
for any \(E, F : \vAbSES(B,A)\).
As such, it suffices to understand paths in \verb+AbSES+\((B,A)\).
These are in turn characterised by Theorem~2.7.2 of loc.\ cit., which characterises paths in general \(\Sigma\)-types, combined with the fact that paths in \verb+AbGroup+ are isomorphisms.
In our case, the result is that paths between short exact sequences correspond to isomorphisms between the \verb+middle+s making the appropriate triangles commute.
We refer to this data as \emph{path data}, and bundle it into a separate type (where \verb+*+ denotes products of types):\href{https://github.com/HoTT/Coq-HoTT/blob/3062f0a152dca5b58323bffb9fceab4188d96bb1/theories/Algebra/AbSES/Core.v\#L80}{\(^\coderef\)}

\begin{lstlisting}[label=list:abses-path-data-iso]
Definition abses_path_data_iso {B A : AbGroup} (E F : AbSES B A)
  := {phi : Iso E F & (phi $o inclusion E == inclusion F)
                      * (projection E == projection F $o phi)}.
\end{lstlisting}
Here \verb+Iso+ forms the type of isomorphisms between two groups.
From our discussion above, for any \(E, F : \vAbSES(B,A)\), we get an equivalence of types%
\href{https://github.com/HoTT/Coq-HoTT/blob/3062f0a152dca5b58323bffb9fceab4188d96bb1/theories/Algebra/AbSES/Core.v\#L104}{\(^\coderef\)}
\[ (E =_{\vAbSES(B,A)} F) \ \simeq \ \verb+abses_path_data_iso+(E,F) . \]
However, a bit more can be said: the \emph{short five lemma}\href{https://github.com/HoTT/Coq-HoTT/blob/3062f0a152dca5b58323bffb9fceab4188d96bb1/theories/Algebra/AbSES/Core.v\#L146}{\(^\coderef\)}
implies that if we replace \verb+Iso+ by \verb+Hom+ above, then it still follows that \verb+phi+ is an isomorphism.
We define \verb+abses_path_data+%
\href{https://github.com/HoTT/Coq-HoTT/blob/3062f0a152dca5b58323bffb9fceab4188d96bb1/theories/Algebra/AbSES/Core.v\#L184}{\(^\coderef\)}
as \verb+abses_path_data_iso+ above, but with \verb+Hom+ in place of \verb+Iso+.
It is convenient to have both types around: it is easier to construct an element of \verb+abses_path_data+; however we will see situations later on where it is convenient to keep track of a \emph{specific} inverse to the underlying map, which \verb+abses_path_data_iso+ lets us do.

\begin{definition}
  The type \(\vAbSES(B,A)\) is a groupoid whose graph structure is given by \verb~abses_path_data_iso~ and a corresponding category structure.
  For the \(2\)-graph structure, we assert that two path data are equal just when their underlying maps are homotopic.%
  \href{https://github.com/HoTT/Coq-HoTT/blob/56629c19010a7d7155b22aad4525f9b2ac1bd584/theories/Algebra/AbSES/Core.v\#L286}{\(^\coderef\)}%
\end{definition}

This definition is justified by the preceding discussion, which yields:

\begin{lemma} \label{lem:abses-path-data}
  For any \(E, F : {\verb+AbSES+}(B,A) \), there are equivalences of types\href{https://github.com/HoTT/Coq-HoTT/blob/3062f0a152dca5b58323bffb9fceab4188d96bb1/theories/Algebra/AbSES/Core.v\#L209}{\(^\coderef\)}
  \[ (E = F) \ \simeq \ {\verb+abses_path_data_iso+}(E, F) \ \simeq \ \verb+abses_path_data+(E,F). \]
\end{lemma}

Though elementary, this lemma has an interesting consequence.
This statement appears as the \(n, i = 1\) case of \cite[Theorem~1]{Retakh1986}.

\begin{proposition} \label{prop:retakh}
  The loop space of \(\verb+AbSES+(B,A)\) is naturally isomorphic to \(\vHom(B,A)\).\href{https://github.com/HoTT/Coq-HoTT/blob/3062f0a152dca5b58323bffb9fceab4188d96bb1/theories/Algebra/AbSES/Core.v\#L481}{\(^\coderef\)}
\end{proposition}

\begin{proof}
  It suffices, by the previous lemma, to give an isomorphism  between \( \vHom(B,A) \) and
   \(\verb+abses_path_data+(A \oplus B, A \oplus B) \).
   One can easily check that a map \(\phi : A \oplus B \to A \oplus B\) subject to the constraints of path data, is uniquely determined by the composite\href{https://github.com/HoTT/Coq-HoTT/blob/3062f0a152dca5b58323bffb9fceab4188d96bb1/theories/Algebra/AbSES/Core.v\#L446}{\(^\coderef\)}
   \[ B \to A \oplus B \xra{\phi} A \oplus B \to A. \]
   Moreover, this association defines a group isomorphism---details are in the formalisation.\href{https://github.com/jarlg/Yoneda-Ext/blob/0d8bfe8e168bbdf325e805d7268e826e889189f0/Lemmas.v\#L28}{\(^\coderef\)}
\end{proof}

To formalise the previous proposition, we first developed basic theory about biproducts of abelian groups which now live in \verb~Algebra.AbGroups.Biproduct~.

In ordinary homological algebra, an abelian group \(P\) is \emph{projective} if for any homomorphism \(f : P \to B\) and epimorphism \(p : A \to B\), there exists a \emph{lift} \(l : P \to A\) such that \( f = e \circ l \).
It is well-known that \(\Ext^1(P,A)\) always vanishes when \(P\) is projective, and that this property characterises projectivity.
In our setting, we define an abelian group \(P\) to be \textbf{projective} if for any homomorphism \(f\) and epimorphism \(p\) as above, there \emph{merely} exists a lift \(l\) such that \(f = l \circ l \).
The propositional truncation makes this into a \emph{property} of an abelian group, and not a structure.
In Coq, we express this as a type-class:\href{https://github.com/HoTT/Coq-HoTT/blob/3062f0a152dca5b58323bffb9fceab4188d96bb1/theories/Algebra/AbGroups/AbProjective.v\#L19}{\(^\coderef\)}

\begin{lstlisting}[label=list:abprojective]
Class IsAbProjective (P : AbGroup) : Type :=
  isabprojective : forall (A B : AbGroup),
    forall (f : Hom P B), forall (e : Hom A B),
    IsSurjection e -> merely (exists l : P $-> A, f == e $o l).
\end{lstlisting}

As in the classical case, projectives are characterised by the vanishing of Ext:

\begin{proposition} \label{prop:projective-iff-ext-vanishes}
  An abelian group \(P\) is projective if and only if 
  \(\vExt^1(P,A) = 0\) for all \(A\).\href{https://github.com/HoTT/Coq-HoTT/blob/3062f0a152dca5b58323bffb9fceab4188d96bb1/theories/Algebra/AbSES/Ext.v\#L118}{\(^\coderef\)}
\end{proposition}

From the induction principle of \(\Zb\) it follows that \(\Zb\) is projective\href{https://github.com/HoTT/Coq-HoTT/blob/3062f0a152dca5b58323bffb9fceab4188d96bb1/theories/Algebra/AbGroups/Cyclic.v\#L56}{\(^\coderef\)} in the sense we defined above.
Consequently \(\vExt^1(\Zb,A) = 0\) for any abelian group \(A\), and we will use this later on.

\begin{remark}
  There is a subtle point related to projectivity that merits discussion.
  Our definition of projectivity only requires the lift \(l\) to \emph{merely} exist (a property), but one could have asked for actual existence (a structure).
  There is no concept of ``mere existence'' in ordinary mathematics, and when translating concepts into HoTT we have to carefully choose to make something a structure or a property. 
  In this case, our definition of projectivity is justified by \cref{prop:projective-iff-ext-vanishes}.
  If we had made projectivity a structure, then not even \(\Zb\) would be projective, which we need it to be.
\end{remark}

\subsection{Ext as a Bifunctor}

Some of the important structure of \(\Ext^1\) is captured by the fact that it defines a \emph{bifunctor}
\( \Ext^1(-,-) : \Ab^{\op} \times \Ab \to \Ab. \)
This means that \(\Ext^1(-,-)\) is a functor in each variable and that the following ``bifunctor law'' holds:
\begin{equation} \label{eqn:bifunctor-law}
  \Ext^1(f,-) \circ \Ext^1(-,g) = \Ext^1(-,g) \circ \Ext^1(f,-).
\end{equation}

We added a basic implementation of bifunctors to the \verb~WildCat~ library for our purposes, asserting the bifunctor law using the \(2\)-graph structure:\href{https://github.com/HoTT/Coq-HoTT/blob/56629c19010a7d7155b22aad4525f9b2ac1bd584/theories/WildCat/Bifunctor.v\#L7}{\(^\coderef\)}
\begin{lstlisting}
Class IsBifunctor {A B C : Type} `{IsGraph A, IsGraph B, Is1Cat C}
  (F : A -> B -> C) := {
    bifunctor_isfunctor_10 : forall a, Is0Functor (F a);
    bifunctor_isfunctor_01 : forall b, Is0Functor (fun a => F a b);
    bifunctor_isbifunctor :
      forall a0 a1 (f : Hom a0 a1), forall b0 b1 (g : Hom b0 b1),
        fmap (F _) g $o fmap (flip F _) f
          $== fmap (flip F _) f $o fmap (F _) g }.
\end{lstlisting}
Here \verb~flip~ is the map which reverses the order of arguments of a binary function.
We note that in order to state the bifunctor law, we only require \verb+F+ to be a \(0\)-functor in each variable.
As such we only include those instances in this class. 

The bifunctor instance of \(\vExt^1\) will come from a bifunctor instance of \(\vAbSES\), so we work with the latter.
First of all, \(\vAbSES : \verb~AbGroup~^{\op} \to \verb~AbGroup~ \to \verb~Type~\) becomes a \(0\)-functor in each variable by pulling back and pushing out, respectively.

\begin{lemma}
  Let \(g : B' \to B\) be a homomorphism of abelian groups.
  For any short exact sequence \( A \to E \to B \), we have a short exact sequence \( A \to g^*(E) \to B' \).\href{https://github.com/HoTT/Coq-HoTT/blob/56629c19010a7d7155b22aad4525f9b2ac1bd584/theories/Algebra/AbSES/Pullback.v\#L12}{\(^\coderef\)}
  Moreover, if \(E\) is trivial, then so is the short exact sequence \(g^*(E)\).\href{https://github.com/HoTT/Coq-HoTT/blob/56629c19010a7d7155b22aad4525f9b2ac1bd584/theories/Algebra/AbSES/Pullback.v\#L208}{\(^\coderef\)}
\end{lemma}

Dually, one can push out a short exact sequence \(A \to E \to B\) along a map \(f : A \to A'\) to get a short exact sequence \(A' \to f_*(A) \to B\).\href{https://github.com/HoTT/Coq-HoTT/blob/56629c19010a7d7155b22aad4525f9b2ac1bd584/theories/Algebra/AbSES/Pushout.v\#L12}{\(^\coderef\)}

We supply careful proofs that pushout and pullback respect composition of pointed maps\href{https://github.com/HoTT/Coq-HoTT/blob/56629c19010a7d7155b22aad4525f9b2ac1bd584/theories/Algebra/AbSES/Pullback.v\#L280}{\(^\coderef\)} and homotopies between maps,\href{https://github.com/HoTT/Coq-HoTT/blob/56629c19010a7d7155b22aad4525f9b2ac1bd584/theories/Algebra/AbSES/Pullback.v\#L387}{\(^\coderef\)} and that pushing out along the identity map gives the pointed identity map.\href{https://github.com/HoTT/Coq-HoTT/blob/56629c19010a7d7155b22aad4525f9b2ac1bd584/theories/Algebra/AbSES/Pullback.v\#L237}{\(^\coderef\)}
These identities could be shown with shorter proofs, however in \cref{sec:fibre-sequence} we will have to prove coherences involving the paths constructed here, and these coherences are simpler to solve when phrased in terms of path data.
In any case, these proofs make \(\vAbSES\) into a \(1\)-functor in each variable.\href{https://github.com/HoTT/Coq-HoTT/blob/56629c19010a7d7155b22aad4525f9b2ac1bd584/theories/Algebra/AbSES/Pullback.v\#L493}{\(^\coderef\)}\href{https://github.com/HoTT/Coq-HoTT/blob/56629c19010a7d7155b22aad4525f9b2ac1bd584/theories/Algebra/AbSES/Pushout.v\#L443}{\(^\coderef\)}

For the bifunctor law we make use of the following proposition, which is remarkably useful for showing that a given extension is a pullback of another one.

\begin{proposition} \label{prop:pullback-characterisation}
  Suppose given the following diagram with short exact rows:
  \[ \begin{tikzcd}
      A \dar["\alpha"] \rar & E' \rar \dar & B' \dar["g"] \\
      A \rar & E \rar & B \period
  \end{tikzcd} \]
If \(\alpha = \id\) then the top row is equal to the pullback of the bottom row along \(g\).\href{https://github.com/HoTT/Coq-HoTT/blob/56629c19010a7d7155b22aad4525f9b2ac1bd584/theories/Algebra/AbSES/Pullback.v\#L91}{\(^\coderef\)}
\end{proposition}

\begin{proof}
  Since the right square commutes, we get a map \(E' \to g^*(E)\) by the universal property of the pullback.
  This map respects the inclusions and projections, and therefore defines a path by \cref{lem:abses-path-data}.
\end{proof}

There is a dual statement for pushouts in which the rightmost map must be the identity.\href{https://github.com/HoTT/Coq-HoTT/blob/56629c19010a7d7155b22aad4525f9b2ac1bd584/theories/Algebra/AbSES/Pushout.v\#L124}{\(^\coderef\)}

\begin{corollary}
  Any diagram with short exact rows as follows yields a path \(f_*(E) = g^*(F)\).\href{https://github.com/HoTT/Coq-HoTT/blob/56629c19010a7d7155b22aad4525f9b2ac1bd584/theories/Algebra/AbSES/BaerSum.v\#L28}{\(^\coderef\)}
  \[ \begin{tikzcd}
      A \dar["f"] \rar & E \dar \rar & B' \dar["g"] \\
      A' \rar & F \rar & B \period
    \end{tikzcd} \]
\end{corollary}

The corollary lets us swiftly show bifunctoriality:

\begin{proposition}
  The binary map \(\vAbSES : \verb~AbGroup~^{\op} \to \verb~AbGroup~ \to \Type \) is a bifunctor.\href{https://github.com/HoTT/Coq-HoTT/blob/56629c19010a7d7155b22aad4525f9b2ac1bd584/theories/Algebra/AbSES/BaerSum.v\#L223}{\(^\coderef\)}
\end{proposition}

\begin{proof}
  Consider a short exact sequence \(A \to E \to B\) along with two homomorphisms \(f : A \to A' \) and \(g : B' \to B \).
  There is an obvious diagram with short exact rows:
  \[ \begin{tikzcd}
      A \dar["f"] \rar & g^*(E) \dar \rar & B' \dar["g"] \\
      A' \rar & f_*(E) \rar & B \period
    \end{tikzcd} \]
  which by the previous corollary yields a path \(f_*(g^*(E)) = g^*(f_*(E))\), as required.
\end{proof}

\begin{remark}
  The results from \cref{ssec:baer-sum} will show that \(\vAbSES\) is an \emph{\(H\)-space}.\href{https://github.com/HoTT/Coq-HoTT/blob/56629c19010a7d7155b22aad4525f9b2ac1bd584/theories/Algebra/AbSES/BaerSum.v\#L214}{\(^\coderef\)}
  Combining this with \cite[Lemma~2.6]{BCFR}\href{https://github.com/HoTT/Coq-HoTT/blob/56629c19010a7d7155b22aad4525f9b2ac1bd584/theories/Homotopy/HSpace/Core.v\#L138}{\(^\coderef\)}, we deduce that \(\vAbSES\) is a bifunctor into pointed types.
  This does not play a role in the rest of this paper, however.
\end{remark}

\subsection{The Baer Sum} \label{ssec:baer-sum}

The \emph{Baer sum} is a binary operation on \(\Ext^1(B,A)\) which makes it into an abelian group.
Given two extensions \(E, F : \Ext^1(B,A)\) their Baer sum is defined as
\[ E + F :\jeq \Delta^* \nabla_* (E \oplus F) \]
where \(E \oplus F\) is the point-wise direct sum, \(\nabla(a,b) :\jeq a_0 + a_1 : A \oplus A \to A\) is the codiagonal map, and \(\Delta(b) :\jeq (b,b) : B \to B \oplus B\) is the diagonal map.

Together with Dan Christensen and Jacob Ender, we have implemented the Baer sum in \href{https://github.com/HoTT/Coq-HoTT/blob/3062f0a152dca5b58323bffb9fceab4188d96bb1/theories/Algebra/AbSES/BaerSum.v}{\verb~Algebra.AbSES.BaerSum~}.
We define this operation on the level of short exact sequences and then descend the operation to the set \(\vExt^1\) by truncation-recursion.\href{https://github.com/HoTT/Coq-HoTT/blob/3062f0a152dca5b58323bffb9fceab4188d96bb1/theories/Algebra/AbSES/BaerSum.v\#L11}{\(^\coderef\)}

\begin{lstlisting}
Definition abses_baer_sum `{Univalence} {B A : AbGroup}
  : AbSES B A -> ABSES B A AbSES B A := fun E F =>
      abses_pullback ab_diagonal
        (abses_pushout ab_codiagonal (abses_direct_sum E F)).
    
Definition baer_sum `{Univalence} {B A : AbGroup}
  : Ext B A -> Ext B A -> Ext B A.
Proof.
  intros E F; strip_truncations.
  exact (tr (abses_baer_sum E F)).
Defined.
\end{lstlisting}
Above, the \verb+strip_truncations+ tactic is a helper for doing truncation-recursion; it lets us assume that both \(E\) and \(F\) are elements of \(\vAbSES(B,A)\) in order to map into the set \(\vExt^1(B,A)\).
We then simply form the Baer sum of \(E\) and \(F\) on the level of short exact sequences before applying \verb+tr+ to the result.

The formalisation that the Baer sum makes \(\vExt^1(B,A)\) into an abelian group closely follows the ``second proof'' of \cite[Theorem~III.2.1]{Mac63}.

\begin{theorem}
  The set \(\vExt^1(B,A)\) is an abelian group under the Baer sum operation.\href{https://github.com/HoTT/Coq-HoTT/blob/3062f0a152dca5b58323bffb9fceab4188d96bb1/theories/Algebra/AbSES/Ext.v\#L50}{\(^\coderef\)}
\end{theorem}

The proof can be done entirely by chaining together equations once the bifunctoriality of \(\vExt^1\) has been established along with its interaction with direct sums.
To illustrate this, we prove that pushouts respect the Baer sum:

\begin{proposition}
  Let \(\alpha : A \to A'\) be a homomorphism of abelian groups.
  For any abelian group \(B\), pushout defines a group homomorphism
  \( \alpha_* : \vExt^1(B,A) \to \vExt^1(B,A'). \)\href{https://github.com/HoTT/Coq-HoTT/blob/3062f0a152dca5b58323bffb9fceab4188d96bb1/theories/Algebra/AbSES/BaerSum.v\#L236}{\(^\coderef\)}
\end{proposition}

\begin{proof}
  Using bifunctoriality of \(\vExt^1\) and naturality of \(\oplus\), we have:
  \begin{align*}
    \alpha_*(E + F) &= \Delta^* (\alpha_* \nabla_* (E \oplus F))
    = \Delta^* (\nabla_* (\alpha_* \oplus \alpha_*)_* (E \oplus F)) \\
                    &= \Delta^* (\nabla_* (\alpha_* E \oplus \alpha_* F))
                      \jeq \alpha_*E + \alpha_* F . \qedhere
  \end{align*}
\end{proof}
Similarly, pullback defines a group homomorphism as well.\href{https://github.com/HoTT/Coq-HoTT/blob/56629c19010a7d7155b22aad4525f9b2ac1bd584/theories/Algebra/AbSES/BaerSum.v\#L251}{\(^\coderef\)}
These results make \(\vExt^1\) into a bifunctor valued in abelian groups.\href{https://github.com/HoTT/Coq-HoTT/blob/56427d24c185e19deae6cf8af0ad80924276ae3f/theories/Algebra/AbSES/Ext.v\#L83}{\(^\coderef\)}

\section{The Pullback Fibre Sequence} \label{sec:fibre-sequence}

\newcommand{\cxfib}{c}

The main goal of this section is to explain and prove the following mathematical result, and to discuss its formalisation%
\href{https://github.com/HoTT/Coq-HoTT/blob/56629c19010a7d7155b22aad4525f9b2ac1bd584/theories/Algebra/AbSES/PullbackFiberSequence.v}{\(^\coderef\)}
along with some applications.

\begin{theorem} \label{thm:pullback-fibre-sequence}
  Let \(A \xra{i} E \xra{p} B\) be a short exact sequence of abelian groups.
  For any abelian group \(G\), pullback yields a fibre sequence:
  \( \vAbSES(B, G) \xra{p^*} \vAbSES(E, G) \xra{i^*} \vAbSES(A, G). \)\href{https://github.com/HoTT/Coq-HoTT/blob/56629c19010a7d7155b22aad4525f9b2ac1bd584/theories/Algebra/AbSES/PullbackFiberSequence.v\#L390}{\(^\coderef\)}
\end{theorem}

In~\cite[Proposition~2.3.2]{CF} we give a different proof of this result via an equivalence between \(\vAbSES(B,A)\) and pointed maps between Eilenberg--Mac Lane spaces.
However, this different proof seems to only work over \(\Zb\) whereas our proof below works for a general ring (though it has only been formalised for \(\Zb\)).

A sequence of pointed maps \(F \xra{i} E \xra{p} B\) is a \textbf{fibre sequence} if \(p \circ i\) is pointed-homotopic to the constant map, and the induced map \(F \to \fib{p}\) is an equivalence.
Any fibre sequence induces a long exact sequence of homotopy groups~\cite[Theorem~8.4.6]{hottbook}:
\[ \cdots \to \pi_n(F) \to \pi_n(E) \to \pi_n(B) \to
  \cdots \to \pi_0(F) \to \pi_0(E) \to \pi_0(B). \]

In the situation of our theorem, it is immediate from functoriality and exactness of \(E\) that \(i^* \circ p^*\) is constant.
Therefore our goal is to show that the induced map \( \cxfib : \vAbSES(B,G) \to \fib{i^*}\) is an equivalence.\footnote{The map \(\cxfib\) is called \verb~cxfib~ in the code.}
We will do this by first constructing a section of \(c\), and then a contraction of the fibres of \(c\) to the values of this section.
A key part of the formalisation is to work with path data instead of actual paths, since the former has better computational properties.
We will simply use \(E = F\) to denote path data, and refer to it as such, in this section.

\begin{lemma}
  Let \(G \to F \to E\) be a short exact sequence.
  Given path data \(p : i^*(F) = \pt \),
  we construct a short exact sequence
  \( G \to F/A \to B \).%
  \href{https://github.com/HoTT/Coq-HoTT/blob/56629c19010a7d7155b22aad4525f9b2ac1bd584/theories/Algebra/AbSES/PullbackFiberSequence.v\#L62}{\(^\coderef\)}
\end{lemma}

\begin{proof}
  The path data \(p\) means that the sequence \(i^*(F)\) splits.
  Thus we can form the cokernel \(F/A\) as in the diagram:
  
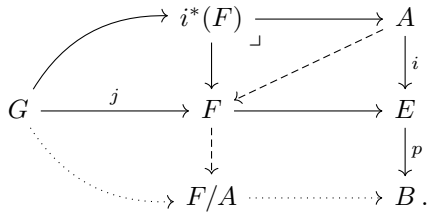
\[ \begin{tikzcd}
      && i^*(F) \ar[rr] \dar \ar[drr, phantom, "\lrcorner" at start] && \ar[dll, dashed] A \dar["i"] \\
      G \ar[urr, bend left] \ar[rr, "j"] \ar[drr, dotted, bend right] && F \ar[rr] \dar[dashed] && E \dar["p"] \\
      && F/A \ar[rr, dotted] && B \period
    \end{tikzcd} \]
  The two maps \(G \to F/A \to B \) are given by composition and the universal property of the cokernel, respectively.
  It is clear that this forms a complex and that the second map is an epimorphism, since it factors one.
  To see that the map \(G \to F/A\) is an injection, suppose \(g : G\) is sent to \(0 : F/A\).
  Then \(j(g)\) is in the image of some \(a : A\) by \(A \to F\).
  But the map \(i^*(F) \to F\) is an injection, being the pullback of one, and so using the path data we get an equality \((g,0) = (0,a)\) in \(G \oplus A\).
  Of course, this implies that \(g = 0\), as required.

  Exactness of \(G \to F/A \to B\) follows from a straightforward diagram chase.
\end{proof}

The diagram above exhibits \(F\) as the pullback of \(F/A\) along \(p^*\), yielding:

\begin{lemma}
  We have path data \( q : p^*(F/A) = F\).\href{https://github.com/HoTT/Coq-HoTT/blob/56629c19010a7d7155b22aad4525f9b2ac1bd584/theories/Algebra/AbSES/PullbackFiberSequence.v\#L135}{\(^\coderef\)}
\end{lemma}

Thus we have given a preimage \(F/A\) of \(F\) under \(p^*\).
To show that the fibre of \(\cxfib\) is inhabited we will show that \(\cxfib(F/A) = (F, p)\), which is a path in \(\fib{i^*}\).
We express all of this in terms of path data, and such a path in \(\fib{i^*}\) then corresponds to path data \(q : p^*(F/A) = F \) which makes the following triangle commute:\href{https://github.com/HoTT/Coq-HoTT/blob/56629c19010a7d7155b22aad4525f9b2ac1bd584/theories/Algebra/AbSES/PullbackFiberSequence.v\#L233}{\(^\coderef\)}
\begin{equation} \label{eq:cxfib-hfiber-path-data}
  \begin{tikzcd}
 i^*p^*(F/A) \ar[dr] \ar[rr, "i^*(q)"] && i^*(F) \ar[dl, "p"] \\
    & G \oplus A
  \end{tikzcd}
\end{equation}
where the rightmost map comes from \(i^*p^*\) being trivial.
The key reason we have formulated things in terms of path data is so that the maps in the triangle above simply compute, because they have all been concretely constructed.

In the following, \(\cxfib\) refers to the map which lands in \(\fib{i^*}\) expressed in terms of path data.\href{https://github.com/HoTT/Coq-HoTT/blob/56629c19010a7d7155b22aad4525f9b2ac1bd584/theories/Algebra/AbSES/PullbackFiberSequence.v\#L150}{\(^\coderef\)}

\begin{lemma}
  We have \(q : \cxfib(F/A) = (F, p)\) in \(\fib{i^*}\).\href{https://github.com/HoTT/Coq-HoTT/blob/56629c19010a7d7155b22aad4525f9b2ac1bd584/theories/Algebra/AbSES/PullbackFiberSequence.v\#L255}{\(^\coderef\)}
\end{lemma}

\begin{proof}
  The previous lemma already yields path data \(q : p^*(F/A) = F\), thus it remains to show that the triangle in \cref{eq:cxfib-hfiber-path-data} commutes.
  The way the maps have been constructed, it's easiest to show this after flipping the triangle so that it starts at \(G \oplus A\) and ends at \(i^*p^*(F/A)\).
  (This is fine since all the maps are isomorphisms.)
  Thus we are comparing two maps out of a biproduct into a pullback.
  To check whether they are equal, we can check it on each inclusion of the biproduct and after projecting out of the pullback.
  In each of these cases one obtains diagrams which commute, but checking this is somewhat involved.
  Fortunately, by our having carefully crafted the path data involved, the maps simply compute and Coq is able to reduce the goal to a simple computation.
\end{proof}

Combining the three previous lemmas, we get a section of \(\cxfib : \vAbSES(B,G) \to \fib{i^*}\).
To conclude that \(\cxfib\) is an equivalence, we contract each fibre over some \((F,p)\) to \((F/A, q)\).

\begin{lemma}
  Suppose \(G \to Y \to B\) is a short exact sequence, and let \(q' : \cxfib(Y) = (F,p)\) in \(\fib{i^*}\).
  Then \((F/A, q) = (Y,q')\) in the fibre of \(\cxfib\) over \((F,p)\).\href{https://github.com/HoTT/Coq-HoTT/blob/56629c19010a7d7155b22aad4525f9b2ac1bd584/theories/Algebra/AbSES/PullbackFiberSequence.v\#L362}{\(^\coderef\)}
\end{lemma}

\begin{proof}
  Under our assumptions, we have the composite map \(\phi : G \oplus A \to i^*p^*(Y) \to p^*(Y)\)\href{https://github.com/HoTT/Coq-HoTT/blob/56629c19010a7d7155b22aad4525f9b2ac1bd584/theories/Algebra/AbSES/PullbackFiberSequence.v\#L285}{\(^\coderef\)} which by a diagram chase can be seen to be the inclusion \(G \to p^*(Y)\) on one component, and \((0,p) : A \to p^*(Y)\) on the other.\href{https://github.com/HoTT/Coq-HoTT/blob/56629c19010a7d7155b22aad4525f9b2ac1bd584/theories/Algebra/AbSES/PullbackFiberSequence.v\#L325}{\(^\coderef\)}.
  Consequently, the composite \(\verb~pr~_1 \circ \phi \circ \verb~in~_A : A \to Y \) is trivial.
  By the universal property of the cokernel, we get an induced map \(F/A \to Y\).
  Once again, by our careful construction of all the maps involved, it is straightforward to simply compute that this map defines path data \(F/A = Y\) and moreover that this path lifts to a path in the fibre of \(\cxfib\).
  There is a coherence between three paths in \(\vAbSES(A,G)\) which is trivially satisfied, since \(\vAbSES(A,G)\) is a \(1\)-type.
\end{proof}

The final lemma implies that the fibres of \(\cxfib\) are contractible, which means that \(c\) is an equivalence and concludes the proof of \cref{thm:pullback-fibre-sequence}.
We now turn our attention to two applications of this theorem.
The first application requires a lemma.

\begin{lemma}
  Let \(g : B' \to B\) be a homomorphism of abelian groups.
  For any \(A\), the following diagram commutes, where the vertical isomorphisms are all given by \cref{prop:retakh}:%
  \href{https://github.com/jarlg/Yoneda-Ext/blob/0d8bfe8e168bbdf325e805d7268e826e889189f0/Lemmas.v\#L104}{\(^\coderef\)}
  \begin{equation}\label{eq:loops-hom}
    \begin{tikzcd}
      \Omega \vAbSES(B,A) \dar["\sim"] \rar["\Omega(g^*)"] & \Omega \vAbSES(B',A) \dar["\sim"] \\
      \vHom(B,A) \rar["\phi \mapsto \phi \circ g"] & \vHom(B',A) \period
    \end{tikzcd}
  \end{equation}
\end{lemma}

\begin{proof}
  Let \(p : A \oplus B = A \oplus B\) be an element of the upper left corner, seen as path data.
  By path induction, one can easily show that the action of \(\Omega(g^*)\) on paths is given by pulling back the path data.
  (Formally, one first proves this for paths with free endpoints, then you can specialise to loops.)
  This means that the following diagram commutes
  \[ \begin{tikzcd}[column sep=large]
      B' \rar \dar[equals]  & A \oplus B' \rar["\Omega(g^*)(p)"] \dar["\id \oplus g"] & A \oplus B' \dar["\id \oplus g"] \rar & A \dar[equals] \\
      B' \rar["{(0,g)}"] & A \oplus B \rar["p"] & A \oplus B \rar & A
    \end{tikzcd} \]
  where we have used the functions underlying the path data \(p\) and \(\Omega(g^*)(p)\), and the unlabeled arrows are the natural ones into or out of a biproduct.
  The composites of the top and bottom rows above are the results of sending \(p\) around the top-right and bottom-left corners of Diagram~\ref{eq:loops-hom}, respectively.
  Since this latter diagram commutes, so does Diagram~\ref{eq:loops-hom}.
\end{proof}

\begin{proposition}[{\cite[Theorem~III.3.4]{Mac63}}] \label{prop:six-term}
  We have an exact sequence of abelian groups:\href{https://github.com/HoTT/Coq-HoTT/blob/832aef3e6fff0f5b953ed170522e1a3d6288a4bb/theories/Algebra/AbSES/SixTerm.v}{\(^\coderef\)}
  \[ \begin{tikzcd}[cramped, column sep=small]
      0 \rar & \vHom(B,G) \rar["p^*"]
      & \vHom(E,G) \rar["i^*"] \ar[d, phantom, ""{coordinate, name=Z}]
      & \vHom(A,G) \ar[dll, rounded corners,
          to path={ -- ([xshift=2ex]\tikztostart.east)
            |- (Z) [near end]\tikztonodes
            -| ([xshift=-2ex]\tikztotarget.west)
            -- (\tikztotarget)}] \\
      & \vExt^1(B,G) \rar["p^*"] & \vExt^1(E,G) \rar["i^*"] & \vExt^1(A,G) \period
      \end{tikzcd} \]
\end{proposition}

\begin{proof}
  This sequence comes from the long exact sequence of homotopy groups~\cite[Theorem~8.4.6]{hottbook} associated to the fibre sequence of~\cref{thm:pullback-fibre-sequence}, using \cref{prop:retakh} and the previous lemma to identify \(\Omega\vAbSES(-,G)\) with \(\vHom(-,G)\).
\end{proof}

\begin{remark}
  The connecting map \(\vHom(A,G) \to \vExt^1(B,G)\) in the sequence above is given by \(\phi \mapsto \phi_* E\).
  Showing this from the fibre sequence is somewhat tedious; we have a proof on paper, but not yet a formalisation.
  Instead, we have formalised a direct proof that the map just stated yields exactness of the sequence.\href{https://github.com/HoTT/Coq-HoTT/blob/832aef3e6fff0f5b953ed170522e1a3d6288a4bb/theories/Algebra/AbSES/SixTerm.v\#L77}{\(^\coderef\)}\href{https://github.com/HoTT/Coq-HoTT/blob/832aef3e6fff0f5b953ed170522e1a3d6288a4bb/theories/Algebra/AbSES/SixTerm.v\#L134}{\(^\coderef\)}
\end{remark}

We apply the six-term exact sequence to compute Ext groups of cyclic groups:

\begin{corollary}[{\cite[Proposition~III.1.1]{Mac63}}] \label{cor:ext-cyclic}
  For any \(n > 0\) and abelian group \(A\), we have\href{https://github.com/HoTT/Coq-HoTT/blob/832aef3e6fff0f5b953ed170522e1a3d6288a4bb/theories/Algebra/AbSES/SixTerm.v\#L222}{\(^\coderef\)}
  \[ \vExt^1(\Zb/n, A) \cong A /n . \]
\end{corollary}

\begin{proof}
  The short exact sequence
  \( \Zb \xra{n} \Zb \to \Zb /n \)
  yields a six-term exact sequence
  \[ \cdots \to \vHom(\Zb, A) \xra{n^*} \vHom(\Zb, A) \to \vExt^1(\Zb/n,A) \to \vExt^1(\Zb,A) \to \cdots \]
  in which the term \(\vExt^1(\Zb,A) \) vanishes since \(\Zb\) is projective.\href{https://github.com/HoTT/Coq-HoTT/blob/56629c19010a7d7155b22aad4525f9b2ac1bd584/theories/Algebra/AbSES/Ext.v\#L118}{\(^\coderef\)}\href{https://github.com/HoTT/Coq-HoTT/blob/56629c19010a7d7155b22aad4525f9b2ac1bd584/theories/Algebra/AbGroups/Cyclic.v\#L56}{\(^\coderef\)}
  This means that the map \(\vHom(\Zb,A) \to \vExt^1(\Zb/n,A)\) is the cokernel of the preceding map.
  By identifying \(\vHom(\Zb,A)\) with \(A\), the claim follows.
\end{proof}

\section{The Long Exact Sequence} \label{sec:les}

We describe our formalisation of the higher Ext groups \(\vExt^n(B,A)\) and their contravariant long exact sequence, which largely follows~\cite[Chapter~III.5]{Mac63}.
The covariant version can be constructed from the arguments in~\cite[Chapter~VII.5]{Mit65}, but we have not formalised this.
The Baer sum is not yet formalised for \(\vExt^n\) \((n > 1)\), so we only have a long exact sequence of \emph{pointed sets}.
Nevertheless, exactness for pointed sets and abelian groups coincide, so we automatically get a long exact sequence of the latter once we have the higher Baer sum.

The formalisation of this section is in the separate repository \href{https://github.com/jarlg/Yoneda-Ext}{\verb~Yoneda-Ext~}, whose \verb~README~ file explains how to set up and build the code related to this chapter.
There are also comments in the code which explain details beyond what we cover here.

\subsection{The Type of Length-\texorpdfstring{\boldmath{\(n\)}}{n} Exact Sequences}

We start by defining a type \(\vES^n\) which we will equip with an equivalence relation by which \(\vExt^n\) will be the quotient.
These constructions will yield functors, which we explain.

The type \(\vES^n(B,A)\) of \textbf{length-\boldmath{\(n\)} exact sequences} is recursively defined as:%
\(\href{https://github.com/jarlg/Yoneda-Ext/blob/0d8bfe8e168bbdf325e805d7268e826e889189f0/ES.v\#L16}{^\coderef}\)
\begin{lstlisting}
Fixpoint ES (n : nat) : AbGroup^op -> AbGroup -> Type
  := match n with
    | 0%nat => fun B A => Hom B A
    | 1%nat => fun B A => AbSES B A
    | S n => fun B A => exists M, (ES n M A) * (AbSES B M)
    end.
\end{lstlisting}
%  \[ \vES^n(B,A) :\jeq
%    \begin{cases}
%      \vHom(B,A) &\text{if } n \jeq 0 \\
%      \vAbSES(B,A) &\text{if } n \jeq 1 \\
%      \Sigma_{C : \verb~AbGroup~} \vES^m(B,C) \times \vAbSES(C,A) & \text{if } n \jeq m+1, m > 0.
%    \end{cases} \]
Thus \(\vES^0(B,A)\) is definitionally \(\vHom(B,A)\), and \(\vES^1(B,A)\) is definitionally \(\vAbSES(B,A)\).
One could also have started the induction at \(n \jeq 1\) instead of \(n \jeq 2\), but it is convenient to have this definitional equality at level \(n \jeq 1\).
The functoriality of \(\vES^n\) is inherited from \(\vAbSES\) and defined in the obvious way by pulling back and pushing out.
For \(n > 0\), an element of \(\vES^{n+1}(B,A)\) is denoted by \((F, E)_M\), with the obvious meaning.
The type \(\vES^n(B,A)^{\href{https://github.com/jarlg/Yoneda-Ext/blob/0d8bfe8e168bbdf325e805d7268e826e889189f0/ES.v\#L32}{\coderef}}\) is pointed by recursion, using the trivial abelian group in the place of \(M\) in the inductive step.

\begin{definition}
  The \textbf{splice} operation is defined as%
  \(^{\href{https://github.com/jarlg/Yoneda-Ext/blob/0d8bfe8e168bbdf325e805d7268e826e889189f0/ES.v\#L144}{\coderef}}\)
  \[ F \splice E :\jeq (F, E)_B : \vES^n(B,A) \to \vAbSES(C,B) \to \vES^{n+1}(C,A). \]
\end{definition}

By induction one can define a general splicing operation in which the second parameter can have arbitrary length%
\(\href{https://github.com/jarlg/Yoneda-Ext/blob/0d8bfe8e168bbdf325e805d7268e826e889189f0/ES.v\#L166}{^\coderef}\),
but we only need the restricted version above.

Now we equip \(\vES^n(B,A)\) with a relation.

\begin{definition}
  We define a relation \(\verb~es_zig~ : \vES^n(B,A) \to \vES^n(B,A) \to \Type\) recursively as follows.
  For \(n = 0, 1\), \verb~es_zig~ is the identity type.
  For \(n \geq 2\), a relation between two elements \((F,E)_M\) and \((Y, X)_N\) consists of a homomorphism \(f : \vHom(M,N)\) along with a path \(f_*(E) = X\) and a relation \(\verb~es_zig~(F, f^*(Y))\) (using functoriality of \(\vES^n\)).%
  \(\href{https://github.com/jarlg/Yoneda-Ext/blob/0d8bfe8e168bbdf325e805d7268e826e889189f0/ES.v\#L186}{^\coderef}\)
\end{definition}

The relation \verb~es_zig~ generates an equivalence relation \verb~es_eqrel~\(\href{https://github.com/jarlg/Yoneda-Ext/blob/0d8bfe8e168bbdf325e805d7268e826e889189f0/ES.v\#L224}{^\coderef}\) (denoted \verb!~! in the code) whose propositional truncation is \verb~es_meqrel~\(\href{https://github.com/jarlg/Yoneda-Ext/blob/0d8bfe8e168bbdf325e805d7268e826e889189f0/ES.v\#L253}{^\coderef}\).
The functoriality of \(\vES^n\) respects these relations.%
\(\href{https://github.com/jarlg/Yoneda-Ext/blob/0d8bfe8e168bbdf325e805d7268e826e889189f0/ES.v\#L264}{^\coderef}
\href{https://github.com/jarlg/Yoneda-Ext/blob/0d8bfe8e168bbdf325e805d7268e826e889189f0/ES.v\#L289}{^\coderef}\)
Basic results on equivalence relations are contained in \href{https://github.com/jarlg/Yoneda-Ext/blob/0d8bfe8e168bbdf325e805d7268e826e889189f0/EquivalenceRelation.v}{\verb~EquivalenceRelation.v~}.

We emphasise that equivalence relation \verb~es_eqrel~ is \emph{not} equivalent to the identity type of \(\vES^n\).
Rather, it is an approximation of the identity type of the classifying space of the category \(\vES^n\) (which we do not know if one can construct in HoTT).
See, e.g.,~\cite[Chapter~III.5]{Mac63} for related discussion.

\begin{definition}
  The pointed set \(\vExt^n(B,A)\) is the quotient of \(\vES^n(B,A)\) by the equivalence relation \verb~es_meqrel~.%
  \(\href{https://github.com/jarlg/Yoneda-Ext/blob/0d8bfe8e168bbdf325e805d7268e826e889189f0/HigherExt.v\#L15}{^\coderef}\)
\end{definition}

The splice operation descends to this quotient.%
\(\href{https://github.com/jarlg/Yoneda-Ext/blob/0d8bfe8e168bbdf325e805d7268e826e889189f0/HigherExt.v\#L96}{^\coderef}\)
By pushing out%
\(\href{https://github.com/jarlg/Yoneda-Ext/blob/0d8bfe8e168bbdf325e805d7268e826e889189f0/HigherExt.v\#L80}{^\coderef}\)
and pulling back%
\(\href{https://github.com/jarlg/Yoneda-Ext/blob/0d8bfe8e168bbdf325e805d7268e826e889189f0/HigherExt.v\#L65}{^\coderef}\)
extensions, \(\Ext^n\) becomes a functor in each variable as well.
Moreover, we have equalities \( f^*(F) \splice E = F \splice f_*(E) \) whenever this expression makes sense, by the definition of \verb~es_zig~.%
\(\href{https://github.com/jarlg/Yoneda-Ext/blob/0d8bfe8e168bbdf325e805d7268e826e889189f0/HigherExt.v\#L140}{^\coderef}\)

\begin{remark}
  The definition of \(\vExt^{n+1}(B,A)\) is, more conceptually, the \((n+1)\)-fold tensor product of functors
  \( \vExt^{n+1}(B,A) = \vExt^n(-,A) \otimes \vExt^1(B,-) \) (see, e.g.,~\cite[Theorem~9.20]{GvBII} or~\cite[Eq.~4.3.4]{Yoneda1960}).
  In our setup, this is a tensor product of Set-valued functors, which can be made into an abelian group by a construction similar to the Baer sum of~\cref{ssec:baer-sum} (though we have not yet formalised this).
  Alternatively, one could define \(\vExt^{n+1}(B,A)\) as the \((n+1)\)-fold tensor product of functors \emph{into abelian groups}.
  \cite[Lemma~2.1]{GvBI} implies that these two definitions coincide.
  We have chosen the present approach because we do not know of a direct construction of the long exact sequence for the latter approach.
\end{remark}

\subsection{The Long Exact Sequence}

We now begin working towards the long exact sequence, following the proof of~\cite[Theorem~XII.5.1]{Mac63}.
As explained at the beginning of this section, we have only formalised the long exact sequence of \emph{pointed sets}---however, exactness for pointed sets is the same as for abelian groups.
Let us first recall the statement:

\begin{theorem} \label{thm:les}
  Let \(A \xra{i} E \xra{p} B\) be a short exact sequence of abelian groups.
  For any abelian group \(G\), there is a long exact sequence by pulling back:%
  \(\href{https://github.com/jarlg/Yoneda-Ext/blob/0d8bfe8e168bbdf325e805d7268e826e889189f0/LES.v\#L15}{^\coderef}
  \href{https://github.com/jarlg/Yoneda-Ext/blob/0d8bfe8e168bbdf325e805d7268e826e889189f0/LES.v\#L47}{^\coderef}
  \href{https://github.com/jarlg/Yoneda-Ext/blob/0d8bfe8e168bbdf325e805d7268e826e889189f0/LES.v\#L94}{^\coderef}\)
  \[ \cdots \xra{i^*} \verb+Ext+^n(A,G) \xra{- \splice E} \verb+Ext+^{n+1}(B,G) \xra{p^*} \verb+Ext+^{n+1}(E,G) \xra{i^*} \cdots \period \]
\end{theorem}

The proof in~\cite{Mac63} first discusses the six-term exact sequence, which we proved as \cref{prop:six-term}.
It then reduces the question to exactness at the domain of the connecting map (Lemma~XII.5.2, loc. cit.), and proves exactness at that spot using Lemmas~XII.5.3, XII.5.4, and XII.5.5.
We will show the three latter lemmas, then directly prove exactness at the other spots, essentially ``in-lining'' Lemma~XII.5.2.

The various constructions we need to do are simpler to carry out on the level of \(\vES^n\) as opposed to \(\vExt^n\).
For this reason we work and formulate things in terms of the former, and then deduce the desired statement for the latter.

Before attacking Lemma XII.5.3, we show the following:

\begin{lemma} \label{lem:5-3-prelim}
  Consider two pairs of short exact sequences which can be spliced:
  \[ (A \xra{l} Y \xra{s} B', \ B' \xra{k} X \xra{r} C), \qquad (A \xra{j} F \xra{q} B, \ B \xra{i} E \xra{p} C). \]
  For any element of \(\verb~es_zig~(Y \splice X, F \splice E)\), we have induced maps
  \( \fib{s_*}(X) \to \fib{q_*}(E) \href{https://github.com/jarlg/Yoneda-Ext/blob/0d8bfe8e168bbdf325e805d7268e826e889189f0/XII_5.v\#L45}{^\coderef} \) and \( \fib{i^*}(F) \to \fib{k^*}(Y) \href{https://github.com/jarlg/Yoneda-Ext/blob/0d8bfe8e168bbdf325e805d7268e826e889189f0/XII_5.v\#L82}{^\coderef} \).
\end{lemma}

\begin{proof}
  We only describe the first map since the second is analogous.
  The zig from \(Y \splice X\) to \(F \splice E\) gives a homomorphism \(f : B' \to B\) along with two paths \(f^*(F) = Y\) and \(f_*(X) = E\).
  Let \( G : \fib{s_*}(X)\); by path induction we may assume \(q_*(G) \jeq X\).
  The path \(f^*(F) = Y\) means we have a commuting diagram:
  \[ \begin{tikzcd}
      A \dar[equals] \rar["l"] & Y \rar["s"] \dar["\phi"] & B' \dar["f"] \\
      A \rar["j"] & F \rar["q"] & B \period
    \end{tikzcd} \]
  Thus \(\phi_*(G)\) defines an element of \(\fib{q_*}(E)\) by
  \( q_*(\phi_*(G)) = f_*(s_*(G)) \jeq f_*(X) = E. \)
\end{proof}

\begin{lemma}[{\cite[Lemma~XII.5.3]{Mac63}}] \label{lem:5-3}
  Given two short exact sequences \(A \xra{j} F \xra{q} B\) and \(B \xra{i} E \xra{p} C\), the following types are logically equivalent:%
  \(\href{https://github.com/jarlg/Yoneda-Ext/blob/0d8bfe8e168bbdf325e805d7268e826e889189f0/XII_5.v\#L173}{^\coderef}\)
  \begin{enumerate}
  \item \(\fib{i^*}(F)\);
  \item \(\fib{q_*}(E)\);
  \item \(\verb~es_eqrel~(\pt, F \splice E)\).
  \end{enumerate}
\end{lemma}

\begin{proof}
  The logical equivalence of between (1) and (2) is as described in~\cite{Mac63}.%
  \(\href{https://github.com/jarlg/Yoneda-Ext/blob/0d8bfe8e168bbdf325e805d7268e826e889189f0/XII_5.v\#L131}{^\coderef}\)
  Moreover, the implication (2) to (3) is clear by the definition of \verb~es_zig~.
  We need to show that (3) implies (1), and we proceed by induction on the length of the zig-zag.

  In the base case we have an actual equality \(\pt = F \splice E \), in which case (1) clearly holds.
  For the inductive step, suppose we have two short exact sequences \( A \xra{l} Y \xra{s} B'\) and \(B' \xra{k} X \xra{r} C \) such that \(Y \splice X\) is related to \(\pt\) by a length \(n\) zig-zag, and we have either zig or a zag relating \(Y \splice X\) to \(F \splice E\).
  If we have a zig, then we use the induction hypothesis to get an element of \(\fib{s_*}(X)\) to which we apply the map
  \( \fib{s_*}(X) \to \fib{q_*} (E) \) from the previous lemma.
  This suffices since (1) and (2) are logically equivalent.

  If we have a zag, then the previous lemma gives a map \(\fib{k^*}(Y) \to \fib{i^*}(F)\), so we are done by the induction hypothesis.
\end{proof}

We reformulate condition (2) in a manner that generalises to \(\vES^n\).%
\(\href{https://github.com/jarlg/Yoneda-Ext/blob/0d8bfe8e168bbdf325e805d7268e826e889189f0/XII_5.v\#L214}{^\coderef}\)

\begin{lstlisting}
Definition es_ii_family `{Univalence} {n : nat} {C B A : AbGroup}
  : ES n.+1 B A -> ES 1 C B -> Type
  := fun E F => { alpha : { B' : AbGroup & B' $-> B }
                       & (es_eqrel pt (es_pullback alpha.2 E))
                         * (hfiber (abses_pushout alpha.2) F) }.
\end{lstlisting}

\begin{lemma}[{\cite[Lemma~XII.5.4]{Mac63}}]
  In the situation of the previous lemma, the types \(\fib{q_*}(E)\) and \(\verb~es_ii_family~(F,E)\) are logically equivalent.%
  \(\href{https://github.com/jarlg/Yoneda-Ext/blob/0d8bfe8e168bbdf325e805d7268e826e889189f0/XII_5.v\#L221}{^\coderef}\)
\end{lemma}

Mac Lane appeals to the six-term exact sequence to prove this lemma, but we give a direct construction.
In order to show Lemma~XII.5.3, we prove a higher analogue of~\cref{lem:5-3-prelim}.
This analogue is phrased in terms of the ``relation fibre'' \verb~rfiber~, which takes the fibre of a point with respect to a relation.

\begin{lemma}
  Let \(n>0\) and consider \(Y : \vES^n(B',A)\), \(F : \vES^n(B,A)\), and two short exact sequences \(B' \xra{k} X \to C \) and \( B \xra{i} E \to C \).
  Given \(\verb~es_zig~(Y \splice X, F \splice E) \), we have maps
  \( \verb~rfiber~_{i^*}(F) \to \verb~rfiber~_{k^*}(Y) \href{https://github.com/jarlg/Yoneda-Ext/blob/0d8bfe8e168bbdf325e805d7268e826e889189f0/XII_5.v\#L285}{^\coderef}\) and \( \verb~es_ii_family~(Y,X) \to \verb~es_ii_family~(F,E) \href{https://github.com/jarlg/Yoneda-Ext/blob/0d8bfe8e168bbdf325e805d7268e826e889189f0/XII_5.v\#L266}{^\coderef}\).
\end{lemma}

\begin{lemma}[{\cite[Lemma~XII.5.5]{Mac63}}]
  Let \(n > 0\), \(F : \vES^n(B,A)\), and \(E : \vES^1(C,B)\).
  The following types are equivalent:%
  \(\href{https://github.com/jarlg/Yoneda-Ext/blob/0d8bfe8e168bbdf325e805d7268e826e889189f0/XII_5.v\#L375}{^\coderef}\)
  \begin{enumerate}
  \item \(\fib{i^*}(E)\);
  \item \(\verb~es_ii_family~(F,E)\);
  \item \(\verb~es_eqrel~(\pt, F \splice E)\).
  \end{enumerate}
\end{lemma}

\begin{proof}
  We first prove an auxiliary lemma which shows that if the three statements are equivalent for a given \(n\), then (1) and (2) are equivalent for \(n+1\).
  The base case for this lemma is simply \cref{lem:5-3}.
  For the inductive step, our auxiliary lemma gives us that (1) and (2) are equivalent.
  It is easy to show that (2) always implies (3), so it remains to show that (3) implies either (1) or (2).
  For this we induct on the length of a zig-zag, and use the equivalence of (1) and (2) along with the previous lemma, similarly (at least in structure) to the proof of \cref{lem:5-3}.
\end{proof}

Afterwards, we reformulate this lemma in terms of \(\vExt^n\).%
\(\href{https://github.com/jarlg/Yoneda-Ext/blob/0d8bfe8e168bbdf325e805d7268e826e889189f0/XII_5.v\#L505}{^\coderef}\)
With this lemma at hand, and using similar methods to the ones presented here, we follow the proof of \cite[Lemma~5.2]{Mac63} to deduce exactness of the long sequence of \cref{thm:les}.

\section{Conclusion}

We have presented a formalisation of the theory of Yoneda Ext in the novel setting of homotopy type theory, starting from the basic definition of a short exact sequence and arriving at the (contravariant) long exact sequence, with various related results along the way.
At present, the long exact sequence is one of pointed sets, and we leave it to future work to formalise the Baer sum on \(\vExt^n\) for \(n > 1\), which would promote this into a long exact sequence of abelian groups.
(The notion of exact sequence coincides for abelian groups and pointed sets.)

For pragmatic reasons we have worked with abelian groups, though it is clear that everything we have done could be applied to general modules.
Even so, the higher Ext groups of abelian groups do not necessarily vanish in HoTT~\cite{CF}, so these are already interesting.
There are various more general approaches that we would like to consider in the future, such as working with \emph{pure} exact sequences (in which the classes of monomorphisms and epimorphisms are appropriately replaced) in an abelian category.

Many of our results have been contributed to the Coq-HoTT library~\cite{coqhott} under the namespace \verb~Algebra.AbSES~, which currently weighs in at about 2900 lines of code (whitespace and comments included).
This excludes the various contributions made to other parts of the library; the precise contributions may be seen through the pull requests
\href{https://github.com/HoTT/Coq-HoTT/pull/1534}{\verb~\#1534~},
\href{https://github.com/HoTT/Coq-HoTT/pull/1646}{\verb~\#1646~},
\href{https://github.com/HoTT/Coq-HoTT/pull/1663}{\verb~\#1663~},
\href{https://github.com/HoTT/Coq-HoTT/pull/1712}{\verb~\#1712~},
\href{https://github.com/HoTT/Coq-HoTT/pull/1718}{\verb~\#1718~},
and \href{https://github.com/HoTT/Coq-HoTT/pull/1738}{\verb~\#1738~}.
In addition, the code for the long exact sequence currently weighs in at about 1350 lines in the separate \href{https://github.com/jarlg/Yoneda-Ext}{\verb~Yoneda-Ext~} repository.

The formalisation covers a substantial part of chapters III.1-3, III.5, and XII.5 of \cite{Mac63}, but also extends beyond the classical theory.
In particular, our proof of \cref{thm:pullback-fibre-sequence} is new even for classical Yoneda Ext (though the theorem is known).
This theorem presented the most challenging part of this formalisation, as it required managing considerable amounts of coherence.
The other challenging part was the long exact sequence, whose proof involves an intricate induction and numerous constructions.
By formalising these theorems we have not only established their correctness but also contributed evidence of the feasibility of dealing with sophisticated mathematical structures in a proof assistant like Coq.

\bibliography{yext}

\end{document}